\documentclass[sn-mathphys,Numbered]{sn-jnl}

\usepackage{graphicx}%
\usepackage{multirow}%
\usepackage{amsmath,amssymb,amsfonts}%
\usepackage{amsthm}%
\usepackage{mathrsfs}%
\usepackage[title]{appendix}%
\usepackage{xcolor}%
\usepackage{textcomp}%
\usepackage{manyfoot}%
\usepackage{booktabs}%
\usepackage{algorithm}%
\usepackage{algorithmicx}%
\usepackage{algpseudocode}%
\usepackage{listings}%
\usepackage[utf8]{inputenc}
\usepackage{amsmath}
\usepackage{verbatim}
\usepackage{amssymb}
\usepackage{listings}
\usepackage{hyperref}
\usepackage{mathtools}
\usepackage{color}
\usepackage{xspace}
\usepackage{lipsum}
\usepackage{tikz-cd}
\usepackage{float}
\usepackage{caption}
\usepackage{subcaption}
\usepackage{amsthm}  
\usepackage{bbm}

\theoremstyle{thmstyleone}%
\newtheorem{theorem}{Theorem}
\theoremstyle{thmstyletwo}%
\newtheorem{remark}{Remark}%

\theoremstyle{thmstylethree}%

\newcommand{\TASEP}{\textsc{Tasep}\xspace}
\newcommand{\forwardmove}{forward\xspace}
\newcommand{\pullbackmove}{pullback\xspace}
\newcommand{\activity}{pointer\xspace} 
\newcommand{\eq}[1]{Eq.~(\ref{#1})}
\newcommand{\eqtwo}[2]{Eqs~(\ref{#1}) and~(\ref{#2})}

\newcommand{\tab}[1]{Table~\ref{#1}}

\newcommand{\lem}[1]{Lemma~\ref{#1}} 
\newcommand{\theo}[1]{Theorem~\ref{#1}}

\newcommand{\OCAL}{\mathcal{O}}  
\newcommand{\alphabar}{\overline{\alpha}}  
\newcommand{\Atilde}{\widetilde{A}}
\newcommand{\Btilde}{\widetilde{B}}

\newcommand{\glb}{\left(}  
\newcommand{\grb}{\right)}  
\newcommand{\glc}{\left[}  
\newcommand{\grc}{\right]}  

\newcommand{\const}{\text{const}}

\newcommand{\TO}{,\ldots,}

\newcommand\bigOb[1]{\ensuremath{\OCAL\glb #1 \grb}}

\newcommand{\fraca}[2]{#1 / #2}

\newcommand{\fracd}[2]{\dfrac{#1}{#2}}

\newcommand{\dSSEP}[2]{d^{\textsc{Ssep}}_{#1, #2}}
\newcommand{\dSSEPwall}[2]{d^{|\textsc{Ssep}|}_{#1, #2}}
\newcommand{\dLTASEP}{d^{\textsc{l-T}}}
\newcommand{\OmegaTASEP}{\Omega^{\textsc{Tasep}}}
\newcommand{\OmegaLTASEP}{\Omega^{\textsc{l-T}}}
\newcommand{\OmegaSSEP}{\Omega^{\textsc{Ssep}}}
\newcommand{\SSEP}{\textsc{Ssep}\xspace}

\newcommand{\ZERO}{\fbox{\phantom{$\bullet$}}}
\newcommand{\ONE}{\fbox{$\bullet$}}
\newcommand{\DOTS}{\fbox{$\!\!\cdots\!\!$}}

\newcommand{\TWO}{\fbox{$\mathrlap{\mkern-4mu \rightarrow}\bullet$}}
\newcommand{\THREE}
{\fbox{$\mathrlap{\mkern-4mu\rightarrow} \mathrlap{\mkern-5mu
\leftarrow}\bullet$}}

\newcommand{\NEG}{\!\!\!\!\!\!\!}
\newcommand{\config}[1]{\{ #1 \}}
\newcommand{\ConfThree}{C}

\newtheorem{conjecture}{Conjecture}
\newtheorem{lemma}{Lemma}

\hypersetup{
	colorlinks=true,
	linkcolor=blue,
	filecolor=magenta,      
	urlcolor=cyan,
	pdftitle={Diameter lifted TASEP},
}

\begin{document}

\title[Diameters of symmetric and lifted simple exclusion
models]{Diameters of symmetric and lifted simple exclusion
models}

\author*[1]{\fnm{Xusheng} \sur{Zhang}}\email{xushengz@psu.edu}

\author[2,3]{\fnm{Werner} \sur{Krauth}}\email{werner.krauth@ens.fr}

\affil*[1]{\orgdiv{Department of Computer Science and Engineering},
\orgname{Pennsylvania State University},
 \orgaddress{\street{W342 Westgate Building, Penn State}, \city{University
Park},
 \postcode{16802}, \state{PA}, \country{USA}}}

\affil[2]{\orgdiv{Laboratoire de Physique de l’Ecole normale sup\'erieure},
\orgname{ENS,  Universit\'e PSL, CNRS, Sorbonne Universit\'e, Universit\'e de
Paris Cit\'e},
\orgaddress{\street{24 rue Lhomond}, \city{Paris}, \postcode{75005},
\country{France}}}

\affil[3]{\orgdiv{Rudolf Peierls Centre for Theoretical Physics},
\orgname{University of Oxford},
\orgaddress{\street{Oxford OX1 3PU}, \country{UK}}}


\abstract{We determine diameters of Markov chains describing one-dimensional
$N$-particle models with an exclusion interaction, namely  the \SSEP (symmetric
simple exclusion process) and one of its non-reversible liftings, the lifted
\TASEP (totally asymmetric simple exclusion process). The diameters provide
lower bounds for the mixing times, and we discuss the implications of our
findings for the analysis of these models.}




\maketitle

\section{Exclusion models}
\label{sec:ExclusionModels}
One-dimensional exclusion models describe $N$ hard-sphere
particles moving on a one-dimensional lattice.
The original model is the symmetric simple exclusion model
(\SSEP)~\cite{Spitzer1970}, which implements a discrete Markov chain based on
the Metropolis algorithm. A configuration $x=\config{x_1, x_2, \dots, x_N}$ in
the sample space $\OmegaSSEP$ is an ordered $N$-tuple with $x_1 < x_2< \dots <
x_N$ of the positions of the $N$ particles, which are indistinguishable. In the
\SSEP, a single move, from time $t$ to time $t+1$, first samples uniformly at
random an active particle $i$ that attempts to move, and then proposes, with
equal probability, a forward or a backward move. If this move would violate the
exclusion condition, it is rejected:
\begin{equation}
\underbrace{\ZERO \ONE \ONE \ZERO \ONE}_{x_t \in \OmegaSSEP}
 \rightarrow
\underbrace{
\ZERO \THREE \ONE \ZERO \ONE}_{\text{choice of  active $i$}}
\rightarrow
\underbrace{ \begin{cases}
\ONE \ZERO \ONE \ZERO \ONE & p= \tfrac12 \\
\ZERO \ONE \ONE \ZERO \ONE & p= \tfrac12
\end{cases} }
_{x_{t+1} \in \OmegaSSEP}.
\label{equ:SSEPDefinition}
\end{equation}
The \SSEP is well-defined both for periodic and hard-wall boundary conditions.

The much studied \TASEP (totally asymmetric simple exclusion
model)~\cite{ChouTASEP2011} differs
from the \SSEP in that, at each time step,
a forward move is proposed:
\begin{equation}
\underbrace{\ZERO \ONE \ONE \ZERO \ONE}_{x_t \in \OmegaTASEP}
 \rightarrow
 \begin{cases}
\ZERO \TWO \ONE \ZERO \ONE \rightarrow \ZERO \ONE \ONE \ZERO \ONE& p = 1/N \\
\ZERO \ONE \TWO \ZERO \ONE \rightarrow \ZERO \ONE \ZERO \ONE \ONE& p = 1/N \\
\underbrace{\ZERO \ONE \ONE \ZERO \TWO}_{\text{choice of active $i$}}
\rightarrow
\underbrace{\ONE \ONE \ONE \ZERO \ZERO}_{x_{t+1} \in
\OmegaTASEP} & p
= 1/N
 \end{cases}.
\label{equ:TASEPDefinition}
\end{equation}
Here, we require periodic boundary conditions because of the fixed direction of
moves.
 The
\TASEP can be interpreted as one half of a lifting~\cite{Chen1999} of the
\SSEP~\cite{Krauth2021eventchain}. In this paper, we consider the lifted
\TASEP~\cite{Essler2023lifted,Lei2019}.
A configuration $x$ in $\OmegaLTASEP$ again consists of an ordered $N$-tuple
$\config{x_1, x_2, \dots, x_N}$  with $x_1 < x_2\dots < x_N$, and in addition
the position $x_i \in x$ of the active particle $i$, which is the only one that
can move:
\begin{align}
\ZERO \ONE \TWO \ZERO \ONE&
 \rightarrow
\ZERO \ONE \ZERO \TWO \ONE
\rightarrow
\begin{cases}
\ZERO \TWO \ZERO  \ONE \ONE & p = \alpha\quad \text{(\emph{pullback} move)}\\
\ZERO \ONE \ZERO \TWO \ONE & p = \alphabar\quad \text{(\emph{forward} move)}
\end{cases}
\label{equ:LTASEPDefNoBlock} \\
\underbrace{\ZERO \TWO \ONE \ZERO \ONE}_{x_t \in \OmegaLTASEP}&
 \rightarrow
\underbrace{
\ZERO \ONE \TWO \ZERO \ONE}_{\text{deterministic}}
\rightarrow
\underbrace{\begin{cases}
\ZERO \TWO \ONE \ZERO \ONE & p = \alpha \quad \text{(pullback move)} \\
\ZERO \ONE \TWO \ZERO \ONE & p = \alphabar \quad \text{(forward move)}
            \end{cases}
}_{x_{t+1} \in
\OmegaLTASEP},
\label{equ:LTASEPDefBlock}
\end{align}
where $\alphabar = 1 - \alpha$.
As sketched in \eqtwo{equ:LTASEPDefNoBlock}{equ:LTASEPDefBlock},
from a sample $x_t$, the active particle $i$ first advances
deterministically, but if
it is blocked ($x_{i+1} = x_i + 1$, with periodic boundary conditions in $i$
and $x$), it
passes the \emph{\activity} on to the blocking particle $i+1$.
Then, the move from $x_t$ to $x_{t+1}$ is completed by transferring the
\activity with probability $\alpha$ to the preceeding particle (\pullbackmove
move from $x_t$ to $x_{t+1}$) or otherwise with probability $\alphabar$, by
taking no further
action (\forwardmove move from $x_t$ to $x_{t+1}$) (see the right-hand sides of
\eqtwo{equ:LTASEPDefNoBlock}{equ:LTASEPDefBlock}). The transition matrix of
this Markov chain is doubly stochastic, so that in the stationary state,
all configurations $x \in \OmegaLTASEP$ are equally
probable~\cite{Essler2023lifted}.

The lifted \TASEP is a one-dimensional model for the event-chain Monte Carlo
algorithm~\cite{Bernard2009} and its parameter $\alpha$ corresponds to a factor
field~\cite{Lei2019}. The model is integrable by Bethe
ansatz~\cite{Essler2023lifted}, although many questions remain. Compared to the
\SSEP, the lifted \TASEP can also be seen as a kinetically constrained lattice
gas~\cite{Cancrini2010}, as each of its configurations is connected only to two
other configurations, whereas in the \SSEP, the connectivity is $2N$.
Nevertheless, we will find that the diameter of the lifted \TASEP is only
roughly double that of the \SSEP. Non-rigorous exact enumerations
for small system sizes indicate that for a critical value of the pullback
$\alpha$, the mixing time scales as $\bigOb{N^2}$ (for $N \propto L$), the same
scaling as we derive here rigorously for the diameter.
In \eqtwo{equ:SSEPDefinition}{equ:TASEPDefinition}, configurations are
represented by the positions of the particles and in
\eqtwo{equ:LTASEPDefNoBlock}{equ:LTASEPDefBlock}, an additional pointer is
introduced. Equivalently, one may also adopt a height representation $H(x)$,
where $H(0)=0$ and
\begin{equation}
 H(x) =
\begin{cases}
H(x-1)& \text{if no particle at $x$}\\
H(x-1) + 1 & \text{else (particle at $x$)}
\end{cases} \quad \forall x \in (1,2,\TO L),
\end{equation}
so that $H(L)=N$. We will however not make use of the height representation.

\section{Diameters of exclusion models}
\label{sec:Diameter}

In a sample space $\Omega$, the diameter is defined  as
\begin{equation}
d = \max_{X,Y\in
\Omega} d(X,Y),
\label{equ:DiameterDefinition}
\end{equation}
where $d(X, Y)$ is the smallest number of moves from $X$ to
$Y$~\cite{Levin2017}. For non-reversible Markov chains, $d(X,Y) $ may differ
from $d(Y,X)$.  The diameter bounds the mixing time from below by $d/2$. The
diameter of a lifted Markov chain cannot be smaller than that of its
\emph{collapsed} chain (the chain of which it is a lifting). Although we are
interested in the diameter of the lifted \TASEP, we therefore first compute that
of the \SSEP. On the other hand, the diameter of a lifted Markov chain can be
much larger than that of its collapsed chain. As an example, the lifted \TASEP
for $\alpha=0$ is not irreducible, so that its diameter is infinite  for finite
$N$ and $L$. Naturally, the diameter of the lifted \TASEP is independent of
$\alpha$, for $0 < \alpha < 1$.

\subsection{Diameter of the SSEP}
We establish the diameter $\dSSEPwall{N}{P}$ of the \SSEP with hard-wall
boundary conditions and, analogously, the diameter $\dSSEP{N}{P}$ for the
periodic \SSEP (see \tab{tab:DiameterSSEP} for examples)
\begin{table}[h]
    \centering
		\begin{tabular}{c|cccccc}
		$L \setminus N$ &1 & 2 &3  &4  &5  &6  \\
			\hline
			4& 2  &2  &2  &0  &  \\
			5&2   &3  & 3 & 2 & 0 \\
			6& 3 & 4 & 5 & 4 & 3 & 0 \\
			7& 3 & 5 & 6 & 6 & 5& 3\\
			8& 4 & 6 & 8 & 8 & 8 & 6\\
			9 & 4 & 7 & 9 & 10 & 10 &9 \\
			10& 5 & 8 & 11 & 12 & 13 & 12\\
			11 & 5 & 9 & 12 & 14 & 15 & 15\\
			12 & 6 & 10 & 14 & 16 & 18 & 18
	   \end{tabular}
    \caption{Diameter $\dSSEP{N}{L}$
    of the \SSEP with periodic boundary conditions, obtained
by exact enumeration. Data are compatible with \eq{equ:SSEPPeriodic}.
}
\label{tab:DiameterSSEP}
\end{table}

\begin{theorem}
\label{thm:diamSSEP}
Let $N \le L$. For hard-wall boundary conditions, we have
\begin{align}
\dSSEPwall{N}{L} &= N(L-N)  \quad \text{(hard-wall boundary conditions)},
\label{equ:SSEPHardWall} \\
\intertext{while, for periodic boundary conditions, the following is true:}
\dSSEP{N}{L}  &=
\left\lceil \fracd{N(L-N)}{2} \right\rceil \quad \text{(periodic
boundary conditions)}.
\label{equ:SSEPPeriodic}
\end{align}
\end{theorem}
Before proving  \theo{thm:diamSSEP}, we
provide an upper bound for the diameter which is independent of boundary
conditions (\lem{lem:SSEPmovex2y}), as well as a tool to reduce
the periodic \SSEP to the hard-wall \SSEP (\lem{lem:interval}).

\begin{lemma}
\label{lem:SSEPmovex2y}
Let $x=\{x_1, x_2, \dots, x_N\}$ and $y=\{y_1, y_2, \dots, y_N\}$
be two configurations in $\OmegaSSEP$. Then, for both periodic and hard-wall
boundary conditions, we have
\begin{equation}
d^{\textsc{Ssep}} (x,y) \le \sum_{i=1}^N |y_i - x_i|.
\end{equation}
Consequently, $\dSSEP{N}{L}, \dSSEPwall{N}{L}\le N(L-N)$.
\end{lemma}

\begin{proof}
For  fixed $L$, we prove by induction on $N\le L$ that $x$ can reach $y$ in at
most
$\sum_{i=1}^N |y_i - x_i|$ moves; these moves do not invoke periodic boundary
conditions, and for each particle $i$, the direction of moves taken from $x_i$
to $y_i$ are  \emph{monotone}, that is, either all forward or all backward. The
case $N=1$ is obvious. Suppose the induction holds for $N\le k$. Now consider $N
= k+1$. We claim that there exists a \emph{flexible particle} at $x_i\in x$ such
that there are no particles in $x$ within the interval $[x_i, y_i]$ or $[y_i,
x_i]$ (whichever is a valid interval), except the flexible particle itself.
Before displacing any other particle, we can move the flexible particle from
$x_i$ to $y_i$ using monotone \SSEP moves only, without invoking periodic
boundary conditions. After this movement, we apply the induction hypothesis to
two smaller instances whose number of particles is at most $k$: we can move
particles in $\config{x_1, \dots, x_{i-1}}$ to $\config{y_1, \dots, y_{i-1}}$
in $
\sum_{j=1}^{i-1} |y_j- x_j|$ monotone moves without invoking periodic boundary
conditions; similarly, we can move particles in $\config{x_{i+1}, \dots,
x_{N}}$ to
$\config{y_{i+1}, \dots, y_{N}}$ in $ \sum_{j=i+1}^{N} |y_j- x_j|$ monotone
moves
without invoking periodic boundary conditions. These two instances are
independent of each other because $y_i$ divides the $[1,L]$ into halves. More
precisely,  if $x_i < y_i$ then $x_{i-1}, y_{i-1} < y_i < y_{i+1}$ and by the
property of flexible particle, $x_{i+1} > y_i$; if $x_i > y_i$ then $x_{i+1},
y_{i+1} > y_i > y_{i-1}$ and by the property of flexible particle, $x_{i-1} <
y_{i}$. Hence, we establish the induction.

It remains to show the existence of a flexible particle. For the sake of
contradiction, we assume that there is no flexible particle. If $y_1
\le x_1$, then the particle $i=1$,  at $x_1$,  is flexible. If the
particle $i=1$, at $x_1$, is not flexible,
then we must have $x_1 < x_2 \le
y_1$. This also implies  $x_2 < y_2$. Moreover, if $x_2$ is not flexible,
then we have $x_2 < x_3 \le y_2$. We continue in this fashion and
obtain for each $i=1,2,\dots, N-1$ that $x_i < x_{i+1} \le y_i$. At the end, we
know $x_N \le y_{N-1} < y_N$. Since, by definition,  particle $N$  is the
most forward, it must be flexible, so there is a contradiction.
\end{proof}

With the following lemma, we reduce the \SSEP with periodic
boundary conditions to the one with hard-wall boundary conditions. We denote by
$[a,b]$ the \emph{periodic interval}  of $a$ and $b$, in other words the
interval from $a$ to $b$ if $a\le b$ and, using this definition, $[a,L] \cup
[1,b]$ if $a>b$. 
\begin{lemma}
\label{lem:interval}
Let $x=\config{x_1, x_2, \dots, x_N}$ and $y=\config{y_1, y_2, \dots, y_N}$ be
configurations
in $\OmegaSSEP$. There then exists a periodic interval $[a,b]$
of length $\lfloor L/2 \rfloor$ or $\lfloor L/2 \rfloor + 1$
that contains the same number of particles in $x$ and in $y$.
\end{lemma}

\begin{proof}[Proof of \lem{lem:interval}]

With periodic boundary conditions, there are $L$ distinct periodic intervals of
length $m := \lfloor L/2 \rfloor + 1$. Let $S_0 = [a_0, b_0]$ be one of the
intervals such that $|x\cap S_0| - |y\cap S_0|$ achieves its maximum among these
$L$ intervals. Let $S_i := [a_i, b_i]$ be the translation of $S_0$ by $i$, where
$a_i:=(a_0 + i) \mod L$ and $b_i:=(b_0 +i) \mod L$. Let $g(i) := |x \cap S_i| -
|y \cap S_i|$. By definition, $g(0)\ge 0$. If $g(0) = 0$, then we are done, so
we assume $g(0)>0$.

Since both $x$ and $y$ have $N$ particles, if $|x \cap S_0| > |y \cap S_0|$,
then there must be some $S_j$ for which $|x\cap S_j| < |y\cap S_j|$, that is,
$g(j) < 0$. For every $i\ge 1$,
\begin{equation}
g(i) - g(i-1)
= \mathbbm{1}[b_{i} \in x] - \mathbbm{1}[b_{i} \in y] -
\mathbbm{1}[a_{i-1} \in x]
+  \mathbbm{1}[a_{i-1} \in y] \in [-2, 2].
\end{equation}
	Among $i=1,\dots, j$, there is the first $j^*$ such that $g(j^*) < 0$. 
	Since $g(j^*) - g(j^*-1) \le -2$ and $g(j^*-1)\ge 0$, $g(j^* - 1) \in
\{1,0\}$.
	If $g(j^* - 1) = 0$, then $S_{j^*-1}$ is the desired interval of length
$m$.
	Otherwise, $g(j^* - 1) = 1$, and $g(j^*) = -1$, which is only possible when
	\begin{equation}
		b_{j^*} \notin x, b_{j^*} \in y, a_{{j^*}-1}\in x, \text{ and }
a_{{j^*}-1}\notin y.
	\end{equation}
	Then we check that $S^*:=[a_{j^*}, b_{{j^*}-1}]$ is a desired interval. 
	Indeed, 
	\begin{align}
		|S^* \cap x| &= |S_{{j^*}-1} \cap x| - 1 
		= g(j^*-1) + |S_{j^*-1} \cap y| - 1 \\
		&= 1 + |S_{j^*-1} \cap y| - 1 = |S_{j^*-1} \cap y| = |S^* \cap y|,
	\end{align}
	and $S^*$ is of length $m-1$.
\end{proof}

\begin{proof}[Proof of \theo{thm:diamSSEP}]
The diameter of the \SSEP is at most $N(L-N)$ for both boundary conditions,
as follows from \lem{lem:SSEPmovex2y}.

For hard-wall boundary conditions,
moving between $\Atilde$ and $\Btilde$,
\begin{equation}
\Atilde = \overset{1}{\ONE} \ONE \DOTS \overset{N}{\ONE} \ZERO  \DOTS \ZERO
\overset{L}{\ZERO}, \quad \Btilde= \ZERO \ZERO
\DOTS \ZERO
\overset{\NEG L-N+1 \NEG}{\ONE} \DOTS
\ONE \overset{L}{\ONE},
\end{equation}
requires $N(L-N)$ moves, because every particle in $\Atilde$ has to be
translated by $L-N$ moves to yield its corresponding particle in  $\Btilde$.
Therefore, \eq{equ:SSEPHardWall} is established, that is, the first claim of
the theorem.

For the periodic \SSEP, we first consider the case of even $L$.
\lem{lem:interval} implies that there exists a periodic
interval $[a,b] \subseteq [1,L]$ of length $L/2$ or $L/2 + 1$ with
$k$ particles both in $x$ and in $y$.
By the periodic boundary conditions, $[1,L] \setminus [a,b]$ is also a
periodic interval $[c,d]$ of length $L/2$ or $L/2-1$.
In both configurations $x$ and $y$, there are $N-k$ particles in $[c,d]$.
Then we obtain two smaller instances of the \SSEP with hard-wall boundary
conditions: on $[a,b]$, $N_1 = k$ and $L_1 = L/2$ (or $L/2 + 1$),
and on  $[c,d]$, $N_2 = N-k$ and $L_2 = L - L_1 = L/2$ (or $L/2 - 1$).
Moreover,
\begin{equation}
	\dSSEP{N}{L} \le
\dSSEPwall{N_1}{L_1} + \dSSEPwall{N_2}{L_2}.
\end{equation}
Two cases must be considered for the length $L_1$.

In the case $L_1 = L_2 = L/2$, we may
apply \lem{lem:SSEPmovex2y} to both instances and obtain:
\begin{align}
\dSSEPwall{N_1}{L_1} + \dSSEPwall{N_2}{L_2} &\le N_1  (L_1 - N_1) + N_2
 (L_2 - N_2) \\
& = k  \left(  \frac{L}{2} - k\right) + (N-k)   \left[
\frac{L}{2}- (N-k)\right]\\
& = \left[k \cdot \frac{L}{2}+ (N-k) \cdot\frac{L}{2}\right] - k^2 -
(N-k)^2 \\
& \le N \cdot \frac{L}{2} - \frac{N^2}{2} =  \left\lceil \frac{N(L-N)}{2}
\right\rceil,
\end{align}
where the last inequality follows from the fact that $\fraca{(x+y)^2}{2} \le
x^2 + y^2$ for any $x,y\in \mathbb{R}$.

Alternatively, in the case $L_1 = L/2 + 1$ and $L_2 = L/2 - 1$, a
similar computation yields for the hard-wall diameters:
\begin{align}
\dSSEPwall{N_1}{L_1} + \dSSEPwall{N_2}{L_2} &\le N_1  (L_1 - N_1) + N_2
 (L_2 - N_2) \\
& = k  \left(  \frac{L}{2} + 1 - k\right) + (N-k)   \left[
\frac{L}{2}- 1 - (N-k)\right]\\
& = \left[k \cdot \frac{L}{2}+ (N-k) \cdot \frac{L}{2}\right]
-
\underbrace{
\glc
 k^2 + (N-k)^2 - k + (N-k)\grc}_{f_1(k)},
\end{align}
where $f_1(k):\mathbb{R} \rightarrow \mathbb{R}$ given by $f_1(k)=k^2 + (N-k)^2
-
k + (N-k)$.
By checking the derivative of $f_1$, we see that for any $k \ge 0$,
\begin{equation}
	f_1(k) \ge f_1\left(\frac{N-1}{2}\right) = \frac{N^2-1}{2}.
\end{equation}
Thus,
\begin{equation}
	\dSSEP{N}{L} \le\dSSEPwall{N_1}{L_1} + \dSSEPwall{N_2}{L_2} \le
	\frac{NL}{2} - f_1(k)
			\le \frac{NL}{2} -  \frac{N^2-1}{2}
			= \frac{N(L-N)}{2} + \frac{1}{2}.
\end{equation}
If $N$ is odd, 
then 
\begin{equation}
	\dSSEP{N}{L} \le 	 \frac{N(L-N)}{2} + \frac{1}{2} = \left\lceil
\frac{N(L-N)}{2} \right\rceil.
\end{equation}
If $N$ is even, then
\begin{equation}
	\dSSEP{N}{L}\le 	 \frac{N(L-N)}{2} + \frac{1}{2} = \left\lceil
\frac{N(L-N)}{2} \right\rceil + \frac{1}{2}.
\end{equation}
Since $\dSSEP{N}{L} $ is an integer and $ \left\lceil \frac{N(L-N)}{2}
\right\rceil $ in this case is the largest integer satisfying the
above inequality, we have
\begin{equation}
	\dSSEP{N}{L} \le  \left\lceil \frac{N(L-N)}{2} \right\rceil.
\end{equation}

When $L$ is odd,
\lem{lem:interval} in turn implies the
existence of a periodic interval $[a,b] \subseteq [1,L]$
of length $\lfloor L/2 \rfloor$ or $\lceil L/2 \rceil$
with the same number of particles in $x$ and in $y$.
In either case, we can choose to set $L_1 = \lfloor L/2 \rfloor, N_1 = k,L_2 =
\lceil L/2 \rceil$ and
$N_2 = N - k$ for some $k\in [0,N]$.
By an analogous computation, we have
\begin{equation}
	\dSSEP{N}{L} \le\dSSEPwall{N_1}{L_1} + \dSSEPwall{N_2}{L_2}
\le \frac{NL}{2} -  \underbrace{\glc  \left(k+\frac{1}{2}\right)k +
(N-k)\left(N-k + \frac{1}{2}\right)\grc}_{f_2(k)},
\end{equation}
with $f_2(k) = \left(k+\frac{1}{2}\right)k + (N-k)\left(N-k-\frac{1}{2}\right)$.
Again, by checking the derivative of $f_2$, we obtain that for any $k\ge 0$, 
\begin{equation}
f_2(k) \ge f_2\left( \frac{2N-1}{4}\right) = \frac{N^2}{2} -\frac{1}{8}.	
\end{equation}
Thus, 
\begin{equation}
	\dSSEP{N}{L} \le  \frac{NL}{2} -f_2(k) = \frac{N(L-N)}{2} + \frac{1}{8}. 
\end{equation}
Since $N(L-N)$ is even,  $\frac{N(L-N)}{2}$ is
the largest integer below $\frac{N(L-N)}{2} + \frac{1}{8}$. 
Since any diameter is an integer, we have effectively
\begin{equation}
	\dSSEP{N}{L}  \le \frac{N(L-N)}{2} = \left\lceil \frac{N(L-N)}{2}
\right\rceil,
\end{equation}
establishing the upper bound of \eq{equ:SSEPPeriodic} for the periodic \SSEP.

Finally we show the lower bound of \eq{equ:SSEPPeriodic} by providing
two configurations that achieve this diameter.
For simplicity, we suppose $N$ and $L$ to be even
(but the following argument generalizes to arbitrary integers  $N$ and $L> N$).
We show that for the periodic \SSEP, moving from configuration $\Atilde$ to configuration $\Atilde_L$, which are shown below,
takes at least  $\lceil \frac{N(L-N)}{2} \rceil$ steps.
\begin{equation}
\Atilde = \overset{1}{\ONE} \ONE \DOTS \overset{N}{\ONE} \ZERO  \DOTS \ZERO
\overset{L}{\ZERO}, \quad \Atilde_L= \ZERO \ZERO
\DOTS \ZERO
\overset{\NEG L/2 + 1 \NEG}{\ONE} \DOTS
\ONE \overset{\NEG L/2 + N \NEG}{\ONE} \ZERO \DOTS \ZERO.
\end{equation}
For a sequence $\mathcal{S}$ of periodic \SSEP moves that lead $\Atilde$ to $\Atilde_L$ 
and for a particle at poition $L/2 + i$ in $\Atilde_L$, we can trace back its
location in $\Atilde$, denoted by $\mathcal{S}_i$.
Crucially, for any $\mathcal{S}$ between $\Atilde$ and
$\Atilde_L$,
there exists a particle $k\in [1,N]$ in $\Atilde_L$ such that
\begin{equation}
	\label{equ:Spart1}
	\mathcal{S}_1 + (k-1) = \mathcal{S}_2 + (k-2)  = \mathcal{S}_3+ (k-3) = \dots = \mathcal{S}_{k}= N,
\end{equation}
and moreover if $k < N$ then 
\begin{equation}
	\label{equ:Spart2}
	1 = \mathcal{S}_{k+1} = \mathcal{S}_{k+2} - 1 =  \dots = \mathcal{S}_{N} - (N - k - 1).
\end{equation}
Accounting for \eqtwo{equ:Spart1}{equ:Spart2},
the number of moves required for $\mathcal{S}$ is at least
\begin{align}
	&\sum_{i=1}^{k} \left[\left(\frac{L}{2} + i\right) - \mathcal{S}_i \right] + 
	\sum_{i=k+1}^{N} \left[\mathcal{S}_i + L - \left(\frac{L}{2} + i\right)  \right]\\
	&= \sum_{i=1}^{k} \left[\left(\frac{L}{2} + i\right) - (N -k + i) \right] + 
	\sum_{i=k+1}^{N} \left[i - k + L - \left(\frac{L}{2} + i\right)  \right] \\
	&= \sum_{i=1}^{k} \left[\frac{L}{2}  - (N-k)  \right] + 
	\sum_{i=k+1}^{N} \left[\frac{L}{2} -k  \right] \\
	& = \frac{NL}{2} - 2k(N-k). \label{eq:Sexp4}
\end{align}  
Note that $k(N-k)$ as a  function of $k$ takes its maximum when $k = \frac{N}{2}$.
Thus, we obtain a lower bound of \eq{eq:Sexp4} 
\begin{equation}
	\frac{NL}{2} - 2k(N-k) \ge \frac{NL}{2} - 2 \left( \frac{N}{2}\right)^2 = 
	\frac{N(L-N)}{2} =  \left\lceil \frac{N(L-N)}{2} \right\rceil.
\end{equation}  
Therefore, $\dSSEP{N}{L}(\Atilde, \Atilde_L) \ge \lceil \frac{N(L-N)}{2} \rceil$, 
and we complete the lower bound of \eq{equ:SSEPPeriodic}.
\end{proof}
\begin{remark}
The proof of \theo{thm:diamSSEP} is constructive. It thus provides
an algorithm for transforming any configuration $x$ of the SSEP with
periodic boundary conditions into  a configuration $y$
in at most $\dSSEP{N}{L}$ \SSEP moves, that first identifies a periodic
interval of appropriate length.
\end{remark}

\subsection{Upper bounds for the diameter of the lifted TASEP}

With the diameter defined as  in  \eq{equ:DiameterDefinition} and applied to
the lifted \TASEP and its sample space $\OmegaLTASEP$,
the following Conjecture~\ref{conj:TASEP} is obtained from exact enumeration for
small
$N$ and $L$ (see \tab{tab:DiameterLTASEP}).
This conjecture suggests that for large $N$ and $L$ with $N/L =
\const$, the diameter of the lifted \TASEP (which is inherently periodic) is
roughly double that of the \SSEP with periodic boundary
conditions.
This is remarkable because the lifted \TASEP has only two possible moves per
(lifted) configuration, whereas the \SSEP has $2N$ moves at its disposal, at
each time step.

\begin{conjecture}
	\label{conj:TASEP}
For $N<L$,
the diameter $\dLTASEP_{N, L}$ of the lifted \TASEP of $N$ particles  on $L$
sites satisfies
\begin{equation}
   \dLTASEP_{N, L} = N (L-N) + 2 N - 3.
\label{equ:LTASEPFamousFormula}
\end{equation}
Furthermore, the configurations $A $ and $B$,
\begin{equation}
A = \overset{1}{\TWO} \ONE  \DOTS \overset{N}{\ONE} \ZERO \ZERO \DOTS
\overset{L}{\ZERO},\quad
B=
\overset{1}{\ZERO} \ZERO \DOTS\ZERO \overset{\NEG L-N+1\NEG}{\ONE} \DOTS
\TWO
\overset{L}{\ONE},
\label{equ:ABDefinition}
\end{equation}
uniquely realize, up to translations, the diameter, so that $\dLTASEP_{N, L} =
d(A, B) $.
\end{conjecture}
\begin{table}[h]
    \centering
		\begin{tabular}{c|cccccc}
		$L \setminus N$
			& 2 &3  &4  &5  &6   \\
			\hline
			5&  7&9 &9 &   &  \\
			6& 9 & 12 & 13 & 12 &    \\
			7& 11 & 15 & 17 & 17 & 15  \\
			8& 13 & 18 & 21 & 22 & 21  \\
			9 & 15 & 21 & 25 & 27 & 27  \\
			10& 17 & 24 & 29 & 32 & 33  \\
			11 & 19 & 27 & 33 & 37 & 39  \\
			12 & 21 & 30 & 37 & 42 & 45
	   \end{tabular}
    \caption{Diameter of the lifted \TASEP,
obtained by exact enumeration. Data are compatible with
\eq{equ:LTASEPFamousFormula}.}
\label{tab:DiameterLTASEP}
\end{table}
We show that the configurations $A$  and $B$ from \eq{equ:ABDefinition}
actually saturate the bound of \eq{equ:LTASEPFamousFormula}:
\begin{lemma}
	\label{lem:moveA2B}
	For $A$ and $B$ in \eq{equ:ABDefinition}, we have
	$\dLTASEP(A,B) = N(L-N) + 2N - 3$.
\end{lemma}
Finally, we establish an upper bound of the diameter for the lifted \TASEP that
is tight up to a small
constant factor:
\begin{theorem}
The diameter of the lifted \TASEP with $N$ particles on $L > N$ sites satisfies
\begin{equation}
\dLTASEP_{N,L} = N(2L-N) + \bigOb{L},
\label{equ:LiftedBound}
\end{equation}
where $\bigOb{L}$ is a function that
increases at most linearly in $L$ and is independent of $N$.
\label{thm:LT}
\end{theorem}

\begin{proof}[Proof of \lem{lem:moveA2B}]
First, we use $N-1$ \forwardmove moves, so that $x_N = N$ becomes active:
\begin{equation}
\label{equ:ar1}
\TWO \ONE \ONE  \ZERO \ZERO \ZERO \to \ONE \ONE \TWO \ZERO \ZERO \ZERO.
\end{equation}
	Next we forward particle $N$ to the position $L-1$, followed by a
pullback move, to make particle $N-1$ active.
\begin{equation}
\ONE \ONE \TWO \ZERO \ZERO \ZERO \to \ONE \ONE \ZERO \ZERO \TWO \ZERO
\to \ONE \TWO \ZERO \ZERO \ZERO \ONE.
\end{equation}
We repeat these moves until any particle $i$ is
at position $L-N+i$. For this, we have used $N(L-N)$ moves:
\begin{equation}
\ONE \TWO \ZERO \ZERO \ZERO \ONE \to \TWO \ZERO \ZERO \ZERO \ONE \ONE
\to \ZERO \ZERO \ZERO \TWO \ONE \ONE.
\label{equ:fr1}
\end{equation}
	Finally, we need $N-2$ additional forward moves to render particle $N-1$
active without moving particles:
\begin{equation}
\ZERO \ZERO \ZERO \TWO \ONE \ONE \to \ZERO \ZERO \ZERO \ONE \TWO \ONE.
\label{equ:ar2}
\end{equation}
In total, $(N-1) + N(L-N) + (N-2) = N(L-N) + 2N - 3$ moves have led from $A$ to
$B$ without invoking periodic boundary conditions.

We have shown that $\dLTASEP(A,B) \le N(L-N) + 2N - 3$ with $A$ and $B$ from
\eq{equ:ABDefinition}. To establish the equality in this formula, we now
argue that each of the above-mentioned moves is necessary. First,
since each move can change at most one particle by one unit in position,  the
required mass
transport demands at least $N(L-N)$ moves from $A$ to $B$.
This justifies the moves in \eq{equ:fr1}. Moreover, starting from $A$, the
only way one can allow any particle to be moved is following the moves as in
\eq{equ:ar1}, which takes $N-1$ moves. Similarly, the only way that a state can
transition to $B$ is through the $N-2$ moves as used in \eq{equ:ar2}. As
\eqtwo{equ:ar1}{equ:ar2} happen respectively before and after the mass
transport, they contain no overlapping moves. Also, moves in both
\eqtwo{equ:ar1}{equ:ar2} are independent of the mass transport
since they change no particle positions. Hence, it takes at
least $N(L-N) + (N-1) + (N-2) = N(L-N) + 2N - 3$ moves to reach $B$ from $A$.
\end{proof}

\begin{proof}[Proof of \theo{thm:LT}]
Let $A$ be as defined in \eq{equ:ABDefinition} and  $B'$ as
shown below:
\begin{equation}
	B' =	\overset{1}{\ZERO} \ZERO \ZERO  \ZERO \overset{\NEG L-N+1\NEG}\ONE
\DOTS \ONE
	\overset{L}\TWO,\quad A = \overset{1}\TWO \ONE \DOTS \overset{N}\ONE \ZERO
\ZERO \ZERO \ZERO.
\end{equation}
We say that a configuration $y'$ is an \emph{ancestor} of a configuration $y$ if
$y'$ can reach $y$ by forward moves that do not cross the periodic boundary
conditions at position $L$. As a forward move advances the pointer by one unit,
a configuration with the pointer at $i$  has  $i-1 \le L-1$ ancestors.
Let $\ConfThree_y$ be an ancestor of $y$ such that the backmost particle in
$\ConfThree_y$ is active:
\begin{equation}
\text{for~~} y= \ZERO\ONE\ONE\ONE \TWO\ZERO\ONE\ZERO,\quad \ConfThree_y =
\ZERO\TWO\ONE\ONE \ONE\ZERO\ONE\ZERO.
\end{equation}
For any pair of configurations $x$ and $y$, there is a sequence of lifted-\TASEP
moves from $x$ to $y$ in four stages:
\begin{equation}
\underbrace{x \to B'}_{\text{stage 1}} \quad
\underbrace{B' \to A}_{\text{stage 2}}  \quad
\underbrace{A \to \ConfThree_y} _{\text{stage 3}}\quad
\underbrace{\ConfThree_y \to y}_{\text{stage 4}}.
\end{equation}
In stage 1, we use moves analogous to those in the proof of \lem{lem:moveA2B}.
Each particle costs at most $L-N$ moves. Also at the beginning and at the
end, $\bigOb{L}$ moves send the \activity to the proper particle.
Hence, there are $N(L-N) + \bigOb{L}$ moves in stage 1:
\begin{align}
	x= &\ZERO\ZERO\ONE\ZERO\TWO\ONE\ZERO\ONE \to
	\ZERO\ZERO\ONE\ZERO\ONE\TWO\ZERO\ONE \to
	\ZERO\ZERO\ONE\ZERO\TWO\ZERO\ONE\ONE \notag \\
	\to
	&\ZERO\ZERO\TWO\ZERO\ZERO\ONE\ONE\ONE \to
	\ZERO\ZERO\ZERO\ZERO\TWO\ONE\ONE\ONE \to
	\ZERO\ZERO\ZERO\ZERO\ONE\ONE\ONE\TWO = B'.
\end{align}
Given the periodic boundary conditions, stage 2 is standard and costs $N^2$
moves.
In stage 3, we forward each particle to its position in $\ConfThree_y$,
and repeat this operation at most $N$ times,
starting from the frontmost particle:
\begin{equation}
\label{eq:A9}
A = \TWO \ONE \ONE \ONE \ZERO \ZERO \ZERO \ZERO \to
\ONE \ONE \ONE \TWO \ZERO \ZERO \ZERO \ZERO \to
\ONE \ONE \TWO \ZERO \ZERO \ZERO \ONE \ZERO
\end{equation}
\begin{equation}
\label{eq:A10}
\to \ONE \TWO  \ZERO \ZERO \ZERO \ONE \ONE \ZERO \to
 \TWO\ZERO\ONE\ZERO\ZERO\ONE\ONE\ZERO  = \ConfThree_y.
\end{equation}
Lastly, in stage 4, by the definition of the ancestor, $\ConfThree_y$ reaches $y$
after at most $L-1$ forward moves. Therefore, $x$ reaches $y$ in
\begin{equation}
N(L-N) + \bigOb{L} + N^2 + N(L-N) + \bigOb{L} + \bigOb{L} = N(2L-N) + \bigOb{L}
\end{equation}
moves, which establishes \eq{equ:LiftedBound}.
\end{proof}

\section{Discussion}
\label{sec:Outlook}

The discussed diameter bounds are compatible with the
proven~\cite{Lacoin2016detailed,Lacoin_2017_SSEP} or
conjectured~\cite{Essler2023lifted} mixing times for the \SSEP and the lifted
\TASEP, respectively. In particular, the diameter of the lifted \TASEP, a first
geometrical characteristic of its sample space, sets a lower bound on its mixing
time, which appears to be sharp. Given the importance of the model, it would be
interesting to obtain other geometrical characteristics and, in first place, the
$\alpha$-dependent conductance of the underlying graph. At present, the
inverse-gap scaling for the lifted \TASEP is known only from the Bethe ansatz,
and it shows a pronounced dependence on $\alpha$. It will be fascinating to
understand whether this $\alpha$ dependence is reflected in the basic geometry
of the sample space and whether it can be obtained rigorously.

\section*{Acknowledgement}
\thanks{We thank A. C. Maggs for helpful discussions. We thank
the mathematical research institute MATRIX in Australia where this research was
initiated.
}


\bibliography{General.bib,LTASEPDiameterArxiv.bib}

\end{document}